\documentclass[aps, prx, amsmath, amssymb, amsfonts, twocolumn,superscriptaddress, footinbib]{revtex4-1}
\usepackage[german,american,english]{babel}
\usepackage{graphicx}
\usepackage{graphics}
\usepackage{dcolumn}
\usepackage{bm}
\usepackage{amssymb}
\usepackage{amsmath}
\usepackage{amsfonts}
\usepackage{epsfig}
\usepackage{mathbbol}
\usepackage{latexsym}
\usepackage{dcolumn}
\usepackage{subfigure}
\usepackage{hyperref}
\usepackage{float}
\usepackage[normalem]{ulem}
\usepackage[usenames, dvipsnames]{color}
\usepackage{epstopdf}

\hypersetup{                    
    colorlinks,
    citecolor=blue,
    filecolor=blue,
    linkcolor=blue,
    urlcolor=blue
}
\usepackage[utf8]{inputenc}

\usepackage{algorithm}
\usepackage[noend]{algpseudocode}
\usepackage{caption}
\usepackage[paperwidth=199.8mm,paperheight=297mm,centering,hmargin=15mm,vmargin=2cm]{geometry}
\usepackage[dvipsnames]{xcolor}
\usepackage{framed}
\definecolor{shadecolor}{rgb}{0.9,0.9,0.9}
\usepackage{mathtools}
\usepackage{amsmath}
\usepackage[shortlabels]{enumitem}
\usepackage[titletoc,title]{appendix}
\usepackage{graphicx,epic,eepic,epsfig,amsmath,latexsym,amssymb,verbatim,color}
\usepackage{dsfont}

\usepackage{float}
\usepackage{tikz}
\usetikzlibrary{chains}
\usetikzlibrary{fit}
\usepackage{pgflibraryarrows}		
\usepackage{pgflibrarysnakes}		

\usepackage{epsfig}
\usetikzlibrary{shapes.symbols,patterns} 
\usepackage{pgfplots}

\usepackage[strict]{changepage}
\usepackage{hyperref}
\hypersetup{colorlinks=true,citecolor=blue,linkcolor=blue,filecolor=blue,urlcolor=blue,breaklinks=true}

\usepackage[marginal]{footmisc}
\usepackage{url}
\usepackage{theorem}
\newtheorem{definition}{Definition}

\newtheorem{theorem}[definition]{Theorem}

\usepackage{colortbl}
\usepackage{pifont}
\definecolor{Gray}{gray}{0.92}
\definecolor{Gray2}{gray}{0.75}
\definecolor{maroon}{cmyk}{0,0.87,0.68,0.32}
\usepackage{booktabs} 

\def\squareforqed{\hbox{\rlap{$\sqcap$}$\sqcup$}}
\def\qed{\ifmmode\squareforqed\else{\unskip\nobreak\hfil
\penalty50\hskip1em\null\nobreak\hfil\squareforqed
\parfillskip=0pt\finalhyphendemerits=0\endgraf}\fi}
\def\endenv{\ifmmode\;\else{\unskip\nobreak\hfil
\penalty50\hskip1em\null\nobreak\hfil\;
\parfillskip=0pt\finalhyphendemerits=0\endgraf}\fi}
\newenvironment{proof}{\noindent \textbf{{Proof~} }}{\hfill $\blacksquare$}

\newcounter{remark}

\newcounter{example}

\mathchardef\ordinarycolon\mathcode`\:
\mathcode`\:=\string"8000
\def\vcentcolon{\mathrel{\mathop\ordinarycolon}}
\begingroup \catcode`\:=\active
  \lowercase{\endgroup
  \let :\vcentcolon
  }

\usepackage{cleveref}
\usepackage{graphicx}
\usepackage{xcolor}

\RequirePackage[framemethod=default]{mdframed}
\newmdenv[skipabove=7pt,
skipbelow=7pt,
backgroundcolor=darkblue!15,
innerleftmargin=5pt,
innerrightmargin=5pt,
innertopmargin=5pt,
leftmargin=0cm,
rightmargin=0cm,
innerbottommargin=5pt,
linewidth=1pt]{tBox}

\newmdenv[skipabove=7pt,
skipbelow=7pt,
backgroundcolor=blue2!25,
innerleftmargin=5pt,
innerrightmargin=5pt,
innertopmargin=5pt,
leftmargin=0cm,
rightmargin=0cm,
innerbottommargin=5pt,
linewidth=1pt]{dBox}
\newmdenv[skipabove=7pt,
skipbelow=7pt,
backgroundcolor=darkkblue!15,
innerleftmargin=5pt,
innerrightmargin=5pt,
innertopmargin=5pt,
leftmargin=0cm,
rightmargin=0cm,
innerbottommargin=5pt,
linewidth=1pt]{sBox}
\definecolor{darkblue}{RGB}{0,76,156}
\definecolor{darkkblue}{RGB}{0,0,153}
\definecolor{blue2}{RGB}{102,178,255}
\definecolor{darkred}{RGB}{195,0,0}

\newcommand{\nc}{\newcommand}
\nc{\rnc}{\renewcommand}
\nc{\beg}{\begin{equation}}
\nc{\eeq}{{\end{equation}}}
\nc{\beqa}{\begin{eqnarray}}
\nc{\eeqa}{\end{eqnarray}}
\nc{\lbar}[1]{\overline{#1}}
\nc{\bra}[1]{\langle#1|}
\nc{\ket}[1]{|#1\rangle}
\nc{\ketbra}[2]{|#1\rangle\!\langle#2|}
\nc{\braket}[2]{\langle#1|#2\rangle}

\nc{\proj}[1]{| #1\rangle\!\langle #1 |}
\nc{\avg}[1]{\langle#1\rangle}
\nc{\rank}{\operatorname{Rank}}
\nc{\smfrac}[2]{\mbox{$\frac{#1}{#2}$}}
\nc{\tr}{\operatorname{Tr}}
\nc{\ox}{\otimes}
\nc{\dg}{\dagger}
\nc{\dn}{\downarrow}
\nc{\cA}{{\cal A}}
\nc{\cB}{{\cal B}}
\nc{\cC}{{\cal C}}
\nc{\cD}{{\cal D}}
\nc{\cE}{{\cal E}}
\nc{\cF}{{\cal F}}
\nc{\cG}{{\cal G}}
\nc{\cH}{{\cal H}}
\nc{\cI}{{\cal I}}
\nc{\cJ}{{\cal J}}
\nc{\cK}{{\cal K}}
\nc{\cL}{{\cal L}}
\nc{\cM}{{\cal M}}
\nc{\cN}{{\cal N}}
\nc{\cO}{{\cal O}}
\nc{\cP}{{\cal P}}
\nc{\cQ}{{\cal Q}}
\nc{\cR}{{\cal R}}
\nc{\cS}{{\cal S}}
\nc{\cT}{{\cal T}}
\nc{\cV}{{\cal V}}
\nc{\cX}{{\cal X}}
\nc{\cY}{{\cal Y}}
\nc{\cZ}{{\cal Z}}
\nc{\cW}{{\cal W}}
\nc{\csupp}{{\operatorname{csupp}}}
\nc{\qsupp}{{\operatorname{qsupp}}}
\nc{\var}{{\operatorname{var}}}
\nc{\rar}{\rightarrow}
\nc{\lrar}{\longrightarrow}
\nc{\polylog}{{\operatorname{polylog}}}
\nc{\wt}{{\operatorname{wt}}}
\nc{\av}[1]{{\left\langle {#1} \right\rangle}}
\nc{\supp}{{\operatorname{supp}}}

\def\x{\xi}

\nc{\RR}{{{\mathbb R}}}
\nc{\CC}{{{\mathbb C}}}
\nc{\FF}{{{\mathbb F}}}
\nc{\NN}{{{\mathbb N}}}
\nc{\ZZ}{{{\mathbb Z}}}
\nc{\PP}{{{\mathbb P}}}
\nc{\QQ}{{{\mathbb Q}}}
\nc{\UU}{{{\mathbb U}}}
\nc{\EE}{{{\mathbb E}}}
\nc{\id}{{\operatorname{id}}}

\nc{\CHSH}{{\operatorname{CHSH}}}

\nc{\be}{\begin{equation}}
\nc{\ee}{{\end{equation}}}
\nc{\bea}{\begin{eqnarray}}
\nc{\eea}{\end{eqnarray}}
\nc{\<}{\langle}
\rnc{\>}{\rangle}
\nc{\rU}{\mbox{U}}

\nc{\ob}[1]{#1}

\nc{\SEP}{{\text{\rm SEP}}}
\nc{\NS}{{\text{\rm NS}}}
\nc{\LOCC}{{\text{\rm LOCC}}}
\nc{\PPT}{{\text{\rm PPT}}}
\nc{\EXT}{{\text{\rm EXT}}}
\nc{\Sym}{{\operatorname{Sym}}}

\nc{\ERLO}{{E_{\text{r,LO}}}}
\nc{\ERLOCC}{{E_{\text{r,LOCC}}}}
\nc{\ERPPT}{{E_{\text{r,PPT}}}}
\nc{\ERLOCCinfty}{{E^{\infty}_{\text{r,LOCC}}}}
\nc{\Aram}{{\operatorname{\sf A}}}

\usepackage{tikz}
\usepackage{hyperref}
\hypersetup{colorlinks=true,citecolor=blue,linkcolor=blue,filecolor=blue,urlcolor=blue,breaklinks=true}

\makeatletter
\def\grd@save@target#1{%
  \def\grd@target{#1}}
\def\grd@save@start#1{%
  \def\grd@start{#1}}
\tikzset{
  grid with coordinates/.style={
    to path={%
      \pgfextra{%
        \edef\grd@@target{(\tikztotarget)}%
        \tikz@scan@one@point\grd@save@target\grd@@target\relax
        \edef\grd@@start{(\tikztostart)}%
        \tikz@scan@one@point\grd@save@start\grd@@start\relax
        \draw[minor help lines,magenta] (\tikztostart) grid (\tikztotarget);
        \draw[major help lines] (\tikztostart) grid (\tikztotarget);
        \grd@start
        \pgfmathsetmacro{\grd@xa}{\the\pgf@x/1cm}
        \pgfmathsetmacro{\grd@ya}{\the\pgf@y/1cm}
        \grd@target
        \pgfmathsetmacro{\grd@xb}{\the\pgf@x/1cm}
        \pgfmathsetmacro{\grd@yb}{\the\pgf@y/1cm}
        \pgfmathsetmacro{\grd@xc}{\grd@xa + \pgfkeysvalueof{/tikz/grid with coordinates/major step}}
        \pgfmathsetmacro{\grd@yc}{\grd@ya + \pgfkeysvalueof{/tikz/grid with coordinates/major step}}
        \foreach \x in {\grd@xa,\grd@xc,...,\grd@xb}
        \node[anchor=north] at (\x,\grd@ya) {\pgfmathprintnumber{\x}};
        \foreach \y in {\grd@ya,\grd@yc,...,\grd@yb}
        \node[anchor=east] at (\grd@xa,\y) {\pgfmathprintnumber{\y}};
      }
    }
  },
  minor help lines/.style={
    help lines,
    step=\pgfkeysvalueof{/tikz/grid with coordinates/minor step}
  },
  major help lines/.style={
    help lines,
    line width=\pgfkeysvalueof{/tikz/grid with coordinates/major line width},
    step=\pgfkeysvalueof{/tikz/grid with coordinates/major step}
  },
  grid with coordinates/.cd,
  minor step/.initial=.2,
  major step/.initial=1,
  major line width/.initial=2pt,
}
\makeatother

\usepackage{thmtools}
\usepackage{thm-restate}
\usepackage{etoolbox}
\makeatletter
\def\problem@s{}
\newcounter{problems@cnt}

\newcommand{\allproblems}{\problem@s}
\makeatother

\definecolor{colorone}{rgb}{1,0.36,0.03}
\definecolor{colortwo}{rgb}{0.54,0.71,0.03}
\definecolor{colorthree}{rgb}{0.01,0.51,0.93}
\definecolor{colorfour}{rgb}{0.47,0.26,0.58}
\usepackage{tcolorbox}
\usepackage{relsize}
\usepackage{graphicx}
\nc{\st}{\text{subject to} \ }
\nc{\supre}{\text{supremum} \ }
\nc{\sdp}{\text{sdp}}
\nc{\cU}{\mathcal U}
\usepackage[qm]{qcircuit}
\usepackage{array}

\begin{document}
\title{Noise-Assisted Quantum Autoencoder}
\author{Chenfeng Cao}
\thanks{chenfeng.cao@connect.ust.hk}
\affiliation{Institute for Quantum Computing, Baidu Research, Beijing 100193, China}
\affiliation{Department of Physics, The Hong Kong University of Science and Technology, Clear Water Bay, Kowloon, Hong Kong, China}

\author{Xin Wang}
\thanks{wangxin73@baidu.com}
\affiliation{Institute for Quantum Computing, Baidu Research, Beijing 100193, China}

\begin{abstract}
Quantum autoencoder is an efficient variational quantum algorithm for quantum data compression. However, previous quantum autoencoders fail to compress and recover high-rank mixed states. In this work, we discuss the fundamental properties and limitations of the standard quantum autoencoder model in more depth, and provide an information-theoretic solution to its recovering fidelity. Based on this understanding, we present a noise-assisted quantum autoencoder algorithm to go beyond the limitations, our model can achieve high recovering fidelity for general input states. Appropriate noise channels are used to make the input mixedness and output mixedness consistent, the noise setup is determined by measurement results of the trash system. Compared with the original quantum autoencoder model, the measurement information is fully used in our algorithm. In addition to the circuit model, we design a (noise-assisted) adiabatic model of quantum autoencoder that can be implemented on quantum annealers. We verified the validity of our methods through compressing the thermal states of transverse field Ising model and Werner states. For pure state ensemble compression, we also introduce a projected quantum autoencoder algorithm. 
\end{abstract}
\date{\today}
\maketitle

\section{Introduction}\label{Introduction}
In classical machine learning, autoencoder is a fundamental tool for learning generative models of data~\cite{Hinton1994,Kramer1991,Vincent2010,Kingma2013,Bengio2013}. Autoencoder learns an encoder map from input data to latent space and a decoder map from latent space to input space, such that the reproduction via the decoder is good for the training set. It has various practical applications, including dimensionality reduction, image denoising, anomaly detection, and information retrieval.

Quantum autoencoder (QAE)~\cite{Romero2017,Wan2017,Verdon2018,Lamata2018} is an approach for quantum data compression. The goal of a quantum autoencoder is to compress the quantum information of initial $(n_{A} + n_{B})$-qubit states onto $n_A$-qubit states via an encoder circuit and then reproduce the initial states approximately via a decoder circuit. Compressing quantum states from a given quantum source is also known as quantum-source coding in quantum Shannon theory. Such dimension reduction of quantum data allows us to perform quantum machine-learning tasks~\cite{Biamonte2017a,Schuld2015a,Arunachalam2017} with reduced quantum resources, and it has potential applications in other quantum technologies. The training procedure of QAE employs the hybrid quantum-classical computation framework and introduces a parameterized variational quantum circuit to learn the encoder. Hybrid quantum-classical algorithms are regarded as well suited for execution on near-term quantum computers~\cite{McClean2016,Endo2020} and they have been widely applied to many topics such as state preparation~\cite{peruzzo2014variational,Yuan2018a,Wang2020a,Chowdhury2020}, quantum linear algebra~\cite{Xu2019a,Huang2019b,bravo2019variational,Wang2020d}, entanglement quantification~\cite{Wang2020c}, quantum distance estimation~\cite{Cerezo2019,Chen2020a}, state and Hamiltonian diagonalization~\cite{Nakanishi2019,larose2019variational,Cerezo2020a,Zeng2020}.

QAE can be seen as an inverse model of quantum error correction \cite{knill1997theory}. In general quantum-error-correction schemes, we encode the information of a small quantum system, e.g., a qubit, to the subspace of a large system, which is resilient to noise. After the system goes through some noise channel, we can recover the state by a decoder. In contrast, QAE starts from a large-system state, compresses it to a small system using an encoder, and then reconstructs the initial state via a decoder. QAE can efficiently denoise certain quantum states \cite{Bondarenko2020,Zhang2020} and has been realized in experiments \cite{Pepper2019,ding2019experimental,huang2020realization}. The compression rate of QAE was also analyzed in~\cite{Ma2020}. More recently, QAE with an additional feature vector to characterize different inputs was proposed in Ref.~\cite{Bravo-Prieto2020}. However, the fundamental limit of the QAE model and its beyond are less understood.

In this paper, we derive theoretical analysis to better understand the limitations of QAE and introduce modified QAE algorithms to go beyond these limitations. First, we show a systematic study of the QAE model and answer what a QAE indeed learns, which helps better understand the fundamental limits of QAE. Second, we present a noise-assisted quantum-autoencoder model that can achieve much higher recovering fidelity than the fundamental limits of original QAE. Compared with the original QAE, we make better use of the measurement information since our noise setup parameters are also determined by the measurement results in the trash system. Noise is a daunting challenge for variational quantum algorithms on current noisy intermediate-scale quantum (NISQ) computers \cite{preskill2018quantum, wang2020noise, zeng2020simulating}, but it plays a positive role in our algorithm. The intuition of our method is that the input and output states of a QAE should have an approximately equal mixedness. When the input is a mixed state with a high rank, we could use adaptive noise channels to increase the mixedness of output. When the input is pure, we could use a projection subroutine to increase the output's purity. Mixedness consistency is an apparent requirement for quantum data compression and recovery. Third, we present a (noise-assisted) adiabatic quantum-autoencoder model for thermal-state compression, where no classical optimization process is needed. Our work may shed light on the application of quantum noise in quantum information processing.

The paper is organized as follows, in Sec. \ref{qae}, we introduce the standard QAE model and give a systematic theoretical analysis. In Sec. \ref{nqae}, we describe our noise-assisted quantum-autoencoder (N-QAE) algorithm. In Sec. \ref{aqae}, we introduce the adiabatic quantum autoencoder (A-QAE) and noise-assisted adiabatic quantum autoencoder (NA-QAE). In Sec. \ref{applications}, we apply our models to compress the thermal states of the transverse-field Ising model and Werner states. The conclusions and future directions are summarized and discussed in Sec. \ref{Conclusion}. Moreover, a projected quantum-autoencoder (P-QAE) algorithm for compressing an ensemble of pure states is discussed in Appendix \ref{pure}.

\section{Quantum Autoencoder}\label{qae}
One of the fundamental tasks in information theory is the compression of information. In general, the problem of data compression is to determine what are the minimal physical resources needed to store an information source. In the quantum world, the task is to compress a quantum source which is described by a Hilbert space $\mathcal{H}$, and a density matrix $\rho$ on that Hilbert space~\cite{Nielsen2010,Watrous2011b,	Wilde2017book,Hayashi2017b}. A compression scheme of rate $R$ for this compression task consists of two families of quantum operations $\cE$ and $\cD$, where $\cE$ is the encoding operation that takes $n$-qubit $\rho$ to a $2^{nR}$-dimensional state and $\cD$ is the decoding operation that maps the compressed states to quantum states in original input space. The behavior of this task is quantified by the fidelity between the original state and the reconstructed state, $i.e.$, $F(\rho,\cD\circ\cE(\rho))$. The scheme is considered reliable if such fidelity approaches one in the limit of $n$. Schumacher’s compression theorem establishes the von Neumann entropy as the fundamental limit on the asymptotic rate of quantum data compression.

Let us consider a bipartite quantum state $\rho_{AB} = \sum_{j}p_j\proj{\psi_j}$ with $n_A$-qubit $A$ subsystem and  $n_B$-qubit $B$ subsystem. The goal of a quantum autoencoder for quantum data compression is to train a parameterized quantum circuit $U$ to compress the input state $\rho_{AB}$ into the $A$ subsystem (latent space), such that one could reproduce $\rho_{AB}$ with high fidelity from the compressed state. The decoder is modeled by a parameterized quantum circuit $V$ acting on the compressed state in $A$ together with a referee state $\ket 0...\ket 0$ in $B'$. $B'$ and $B$ have the same dimension. Usually we choose $V = U^{\dagger}$. The general scheme of a quantum autoencoder is presented in Fig.~\ref{fig:autoencoder}.

The system $B$ is called the trash since it is discarded after the encoding procedure. The task of quantum data compression is pointed out to be related to decoupling in \cite{Romero2017}, in the sense that the trash system $B$ can be perfectly decoupled from the whole system, then the autoencoder can reach lossless compression. Based on this fact, Romero $\textit{et al.}$ \cite{Romero2017} proposed the trash-state cost function that quantifies the degree of decoupling,
\begin{align}
L_d(\bm\theta) :=&1 - F(\proj{0}_B,\rho_{B}^{out}) \label{eq:trash state local function}\\
= & 1 - \tr \proj{0}_B \rho_{B}^{out}\\
= & 1 - \tr (I_A\otimes\proj{0}_B)U(\bm\theta) \rho_{AB}U(\bm\theta)^\dagger,
\end{align}
where $\proj{0}_B$ is the reference state and $\rho_{B}^{out} = \tr_A U(\bm\theta) \rho_{AB}U(\bm\theta)^\dagger = \sum_{j}p_j\tr_A U(\bm\theta) \proj{\psi_j}U(\bm\theta)^\dagger$ is the trash state after applying the encoding circuit. Throughout the paper, we use the fidelity  $F(\rho,\sigma)=(\operatorname{Tr}\sqrt{\sigma^{1/2}\rho\sigma^{1/2}})^2$. When $\rho = |\psi \rangle \langle \psi|$ is a pure state, $F(\rho,\sigma)= \langle\psi| \sigma |\psi\rangle$. Note that the normal fidelity is defined as $\widehat F(\rho,\sigma)=\operatorname{Tr}\sqrt{\sigma^{1/2}\rho\sigma^{1/2}}$.

\begin{figure}[H]
	\centering
	\includegraphics[width=8.5cm]{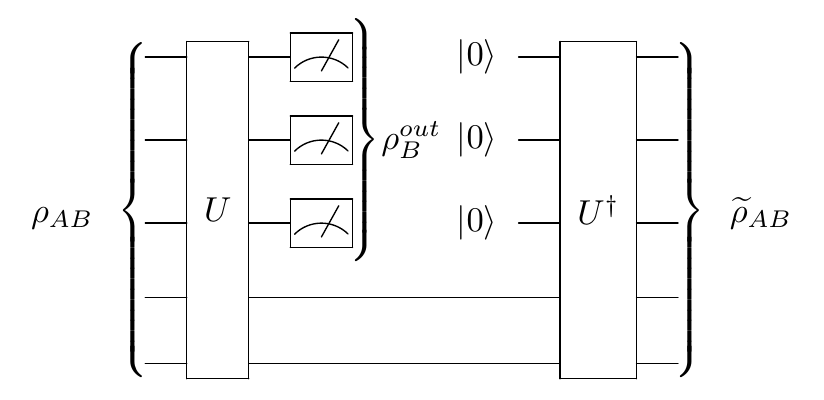}
	\caption{Structure of 5-2-5 QAE. $U$ is a parameterized quantum circuit with parameters $\bm\theta$. After applying $U$, we implement measurements on $\rho_{B}^{out}$ to update $\bm\theta$. $\rho_{A}^{out}=\operatorname{Tr}_B U\rho_{AB}U^\dagger$ is the compressed state. The initial state of $B'$ is $\ket 0...\ket 0$
	}
	\label{fig:autoencoder}
\end{figure}

In the task of quantum data compression, the ability of a quantum autoencoder is quantified by the fidelity of reconstruction:
\begin{align}
F(\rho_{AB},V(\rho_{A}^{out} \otimes \proj 0)V^\dagger),
\end{align}
where $\rho_{A}^{out} = \tr_B U(\bm\theta)\rho_{AB}U(\bm\theta)^\dagger$ is the compressed state and $V$ is the decoder. The corresponding infidelity cost function is then defined by
\begin{align}\label{eq:overall cost}
L_f(\bm\theta, V):= 1 - F(\rho_{AB},V(\rho_{A}^{out} \otimes \proj 0)V^\dagger).
\end{align}

To figure out the features that quantum autoencoders learn from the quantum data, we consider the ideal case where the encoder $U$ is the optimal solution that achieves the minimum of the decoupling cost function in Eq.~\eqref{eq:trash state local function}. That is, $U$ is the optimal solution to 
$\max_{U}\text{ }\langle0|_B(\operatorname{Tr}_A U\rho_{AB}U^{\dagger})|0\rangle_B 
= \max_{U} \operatorname{Tr} U\rho_{AB}U^{\dagger}(I_A\otimes |0\rangle \langle 0|_{B})$.

Previous quantum autoencoders~\cite{Romero2017,Verdon2018,Lamata2018} usually set the decoder as the inverse of the decoder, $i.e.$, $V(\bm\theta) = U(\bm\theta)^\dagger$. Via this step, training via the trash-state cost function~\eqref{eq:trash state local function} will become the minimization of a lower bound to the cost function corresponding to reconstruction fidelity since 
\begin{align*}
L_f(\bm\theta) &= 1 - F(\rho_{AB},V(\bm\theta)(\rho_{A}^{out} \otimes |0\rangle \langle 0|)V(\bm\theta)^\dagger) \\
&= 1- F(V(\bm\theta)^\dagger\rho_{AB}V(\bm\theta),\rho_{A}^{out} \otimes |0\rangle \langle 0|) \\
&= 1- F(U(\bm\theta)\rho_{AB}U(\bm\theta)^\dagger,\rho_{A}^{out} \otimes |0\rangle \langle 0|) \\
& \ge 1- F(\rho_{B}^{out},|0\rangle \langle 0|) = L_d(\bm\theta).
\end{align*}
Suppose the input state $\rho_{AB}$ has the spectral decomposition $\rho_{AB} = \sum_{j=1}^k p_j\proj{\psi_j}$ with decreasing spectrum $\{p_j\}_{j=1}^k$. After fully training, $L_d(\bm\theta) = 1-\sum_{j=1}^{d_A}p_j$~\cite{Ma2020}. We note that $L_d(\bm\theta)$ is only the infidelity between the final state of the trash system $B$ and the initial state of $B'$. However, this is different from the recovering infidelity on our focus. Specifically, we study the infidelity between the initial state and the recovered state, which is denoted by $L_f(\bm\theta)$ and characterizes the fundamental limits of QAE in compressing and recovering the quantum data. In Theorem~\ref{prop:lower bound QAE}, we show that there is always a nonzero gap between $L_d(\bm\theta)$ and $L_f(\bm\theta)$ when $k > d_{A}$.

Directly choosing the unitary $U(\bm\theta)^\dagger$ as the decoder also has drawbacks. Consider the rank of the recovered state $\widetilde\rho_{AB}=V(\bm\theta)(\rho_{A}^{out} \otimes |0\rangle \langle 0|)V(\bm\theta)^\dagger$, denote the dimensions of subsystems $A$ and $B$ as $d_{A}$ and $d_{B}$, we have 
\begin{align}\label{eq:rank constraint}
\rank (\widetilde \rho_{AB})  = \rank (\rho_{A}^{out}) \le d_{A}.
\end{align}
This fact significantly limits the ability of compressing high rank states. 

Before giving the theorem of fidelity bound, we introduce the quantum marginal problem, which is related to the achievability of the bound.
To be specific, given a set of local density matrices (quantum marginals) $\{\rho_{j}\}$, the Quantum Marginal Problem is to determine whether there exists an $d$-dimensional state $\rho$ with spectrum $\{\lambda\}$, such that $\{\rho_{j}\}$ are reduced density matrices of $\rho$.
Quantum marginal problem for general case is QMA-complete, which is unlikely to be efficiently solvable even with a quantum computer~\cite{liu2006consistency}. 

In the following, we give a theoretical bound of the QAE reconstruction fidelity and a information-theoretic proof. 

\begin{theorem}\label{prop:lower bound QAE}
	Suppose the input state $\rho_{AB}$ has the spectral decomposition $\rho_{AB} = \sum_{j=1}^k p_j\proj{\psi_j}$ with decreasing spectrum $\{p_j\}_{j=1}^k$. The reconstruction fidelity for a quantum autoencoder with fully trained encoder $U$ and decoder $U^\dagger$ has an upper bound $F \leq \sum_{j=1}^{d_A}p_j$. 
	\begin{enumerate}
	    \item When $k \leq d_{A}$, the bound can be achieved, $i.e.$, $F = \sum_{j=1}^k p_j = 1$. 
	    \item When $k > d_{A}$, the bound is achievable if and only if there exists a $d_{A}(d_{B}-1)$-dimensional quantum state with spectrum $\{p_i/ \sum_{j=d_{A}+1}^{k}p_j\}_{i=d_{A}+1,...,k}$ that has a $d_{A}$-dimensional quantum marginal with spectrum $\{p_i/ \sum_{j=1}^{d_{A}}p_j\}_{i=1,...,d_{A}}$. Even if the bound is achievable, the probability of achieving this fidelity bound for random input is 0.
	\end{enumerate}

\end{theorem}
\begin{proof}
	First, we will show a fundamental limit on the reconstruction fidelity.
	As pointed out in Eq.~\eqref{eq:rank constraint}, the decoder $U^\dagger$ can only output quantum states with maximum rank $d_A$. Let us assume that the reconstructed state $\widetilde \rho_{AB}$ has the spectral decomposition $\widetilde \rho_{AB} = \sum_{j=1}^{d_A} q_j\proj{\phi_j}$ with decreasing spectrum $\{q_j\}_{j=1}^{d_A}$. Consider the positive operator-valued measure $\{E_j = \proj{\psi_j}\}_{j=1}^{d_Ad_B}$ with orthonormal vectors $\ket{\psi_j}$, we have 
	\begin{align}
	F(\rho_{AB},\widetilde\rho_{AB}) \le & \left(\sum_{j=1}^{d_Ad_B} \sqrt{\tr(\rho_{AB}E_j)\tr(\widetilde\rho_{AB}E_j)} \right)^2\\
	&= \left(\sum_{j=1}^k  \sqrt{p_j \bra{\psi_j}\widetilde\rho_{AB}\ket{\psi_j}}\right)^2 \\
	&\le  \left(\sum_{j=1}^{d_A} \sqrt{p_jq_j} \right)^2\\
	&\le  \left(\sum_{j=1}^{d_A} p_j \right) \left(\sum_{j=1}^{d_A} q_j \right)= \sum_{j=1}^{d_A}p_j.
	\end{align}
	The first inequality follows due to fact that the fidelity is no larger than the fidelity induced by a measurement~\cite{Nielsen2010}.
	The second inequality follows due to the rearrangement inequality and the decreasing spectrum of $\widetilde \rho_{AB}$. The last inequality follows due to the Cauchy Schwarz inequality.
	
	Second, we show the achievability by introducing a feasible protocol. To achieve the maximum overlap with $\proj0_B$, there exists an encoder unitary $U$ transforms $\ket{\psi_j}$ to $\ket{v_j}_A\otimes\ket 0_B$ for each $1\le j\le d_A$ with orthonormal vectors $\{\ket{v_j}_A\}_{j=1}^{d_A}$. Simultaneously, the encoder $U$ transforms $\ket{\psi_j}$ for $d_A<j\le k$ to the vectors in the orthogonal complement of $\rm{span}\{\ket{v_j}_A\otimes\ket 0_B\}_{j=1}^{d_A}$.
	Without loss of generality, we assume that $U$ transforms $\ket{\psi_j}$ for $d_A<j\le k$ to $\ket{w_j}$ such that $\bra{0}_{B}\tr_{A}\proj {w_j}0 \rangle_{B}=0$. The orthogonal complement of
	$\rm{span}\{\ket{v_j}_A\otimes\ket 0_B\}_{j=1}^{d_A}$ 
	is also the space spanned by all possible $\ket{w_j}$, denoted by $\mathcal{H}_A \otimes \mathcal{H}_{\bar{B}}$, where $\mathcal{H}_{\bar{B}}$ is the 
    orthogonal complement of span$\{\ket 0_B\}$ in $\mathcal{H}_B$. 
    $\operatorname{dim} \mathcal{H}_A \otimes \mathcal{H}_{\bar{B}} = d_{A}(d_{B} - 1)$.

	For the behavior of the decoding step, our analysis considers two cases. 
		When $k\le d_A$, the compressed state is given by
		\begin{align}
		\rho_A^{out} 
		&= \sum_{j=1}^{k}p_j\tr_B\proj{v_j}\otimes\proj{0}_B\\
		&= \sum_{j=1}^{k}p_j\proj{v_j}
		\end{align}
		Then the decoder unitary $U^\dagger$ can perfectly recover the compressed state, since it can transform the eigenstates of $\rho_A^{out}\otimes \proj{0}$ back to the eigenstates of $\rho_{AB}$ whose corresponding eigenvalues are encoded in $\rho_A^{out}$, $i.e.$,
		\begin{align}
		U^\dagger(\ket{v_j}_A\otimes\ket 0_B) = \ket {\psi_j}, \quad\forall j.
		\end{align}
		Therefore, we have
		$U^\dagger(\rho_{A}^{out}\otimes \proj0_{B'})U = \sum_{j=1}^{k}p_j\proj{\psi_j} = \rho_{AB}$.
		
		When $k > d_A$, the compressed state is
		\begin{align}
		\rho_A^{out} &= \sum_{j=1}^{d_A}p_j\proj{v_j}+ \sum_{j=d_A+1}^{k}p_j\tr_{\bar{B}}\proj{w_j}.
		\end{align}
		
		Given the quantum marginal in $\mathcal{H}_A$ 
		\begin{equation}
		    \rho_{m}=\sum_{j=1}^{d_A}p_j\proj{v_j}  /(\sum_{j=1}^{d_A}p_j),
		\end{equation}
		if the condition in Theorem \ref{prop:lower bound QAE}(2) is satisfied, there exist $\ket{w_j}$s such that $\rho_{m}$ is the reduced density matrix of a quantum state $\rho_{M} \in \mathcal{H}_A \otimes \mathcal{H}_{\bar{B}}$, i.e.,
		\begin{equation}
		    \rho_{m} = \tr_{\bar{B}} \rho_{M}.
		\end{equation}
		with
		\begin{equation}
		    \rho_{M} = \sum_{j=d_{A}+1}^{k} p_{j} \proj{w_j}/(\sum_{j=d_{A}+1}^{k}p_j).
		\end{equation}
		Note that constraints on the spectrum of quantum marginals are discussed in Refs.~\cite{han2005compatibility, klyachko2004quantum, klyachko2006quantum}. 
		
		In this case, 
		\begin{align}
		\rho_A^{out} & = \sum_{j=1}^{d_A}p_j\proj{v_j}+\frac{\sum_{j=d_{A}+1}^{k}p_j}  {\sum_{j=1}^{d_A}p_j}\sum_{j=1}^{d_A}p_j\proj{v_j}\\
		& = 
		\frac{\sum_{j=1}^{d_A}p_j\proj{v_j}}  {\sum_{j=1}^{d_A}p_j}.
		\end{align}

		Then the reconstructed state is given by
		\begin{align*}
		\widetilde \rho_{AB} = U^{\dagger}(\rho_A^{out}\otimes\proj{0}_{B'})U
        = \frac{\sum_{j=1}^{d_A}p_j\proj{\psi_j}}  {\sum_{j=1}^{d_A}p_j}.
		\end{align*}
        Hence the fidelity of reconstruction is given by
		$F(\widetilde\rho_{AB},\rho_{AB})= \sum_{j=1}^{d_A}p_j$. However, since we cannot choose $\ket{w_j}$ manually in the QAE scheme, the Lebesgue measure of states $\{\rho_M\}$ is 0 in $\mathcal{H}_A \otimes \mathcal{H}_{\bar{B}}$, therefore for random input $\rho_{AB}$, the probability of achieving this fidelity is 0.
		
		If the condition in Theorem \ref{prop:lower bound QAE}(2) is not satisfied, one could only obtain
		\begin{align*}
		\widetilde \rho_{AB} &= \sum_{j=1}^{d_A}p'_j\proj{\psi_j},
		\end{align*}
		where $\{p'_j\} \neq \{p_j/(\sum_{j=1}^{d_A}p_j)\}$, which indicates that the fidelity of reconstruction
		$F(\widetilde\rho_{AB},\rho_{AB}) < \sum_{j=1}^{d_A}p_j$.

\end{proof}

In brief, when $k \leq d_A$, standard QAE can perfectly recover the input state after training; when $k > d_A$, the recovering fidelity is bounded by $\sum_{j=1}^{d_A}p_j$, this bound is not always achievable, its achievability depends on the specific input.

QAE learns to decouple the correlation between subsystems $A$ and $B$ through increasing the fidelity between $\rho_{B}^{out}$ and a target pure state ($i.e.$ $\proj 0$). If $A$ and $B$ are fully decoupled, replacing $\rho_{B}^{out}$ by $\proj 0$ will not destroy the internal entanglement, the decoder can recover the input precisely. However, in this model, the rank of the output state is limited by the dimension of $A$, $d_{A}$. If the input state rank is greater than $d_{A}$, even the optimal encoder can only decouple the eigenstates for the $d_A$ largest eigenvalues. This limitation is due to the underutilization of the information (measure results) of the trash system after training.

\section{Noise-assisted Quantum Autoencoder}\label{nqae}
Inspired by the above fundamental limit on quantum autoencoders, we know that the fidelity of reconstruction via the above quantum autoencoder could be bad for quantum states with rank larger than $d_A$ (dimension of the compressed system), in particular for the states with flat spectrums. For example, let us consider a rank-four three-qubit state $\rho_{AB}$ with spectrum $\{p_j=1/4\}_{j=1}^4$, where $d_A=2$ and $d_B=4$. Based on the above proposition, even for the best case of training, the fidelity of reconstruction is always no larger than ${1}/{2}$.

To overcome the above weakness of only reconstructing a low-rank state with low fidelity, we propose the noise-assisted quantum autoencoder for quantum data compression, which uses quantum noise to assist the decoder in enhancing the fidelity of recovery.  

For the trash system $B$, we implement variational quantum diagonalization \cite{cerezo2020variational} to extract the spectra of $\rho_{B}^{out}$. For a loss Hamiltonian summed by local $Z$ Pauli operators,
\begin{equation}
H_{cost}=\mathbb{1}-\sum_{j=1}^{n_{B}} r_{j} Z_{j}.
\end{equation}
With appropriate choice of $\{r_j\}$ (e.g. $r_{j} = 1/2^{j-1}$), we can ensure that all diagonal terms (eigenvalues) of $H_{cost}$ are nondegenerate and arranged in the ascending order. We optimize the parameters in the encoder $U$ to minimize the cost function 
\begin{equation}
L = \tr[\rho_{B}^{out} H_{cost}],
\end{equation}
only Pauli-$z$ local measurements are required, therefore for shallow quantum circuit, we can avoid the trainability problem~\cite{cerezo2021cost}.  

Instead of inputting the pure state $\ket 0...\ket 0$, we input a mixed state $\rho_{B'}^{in}$ for $B'$ to the decoder to improve the rank of the final-state density matrix. More specifically, we input the mixed state
\begin{equation}
    \sigma_{j} = \begin{pmatrix}
 1-\epsilon_{j}& 0\\ 
 0& \epsilon_{j}
\end{pmatrix}
\end{equation}
to the $j$th qubit in $B'$. $\rho_{B'}^{in}=\otimes_{j=1}^{n_{B}} \sigma_{j}$. Note that the noise rates $\{\epsilon_{j}\}_{j=1}^{n_{B}}$ are determined by measurements on trash system $B$ as follows
\begin{equation}
\epsilon_{j} = \frac{1}{2}-\frac{1}{2}\tr[\rho_{B}^{out} Z_{j}],
\end{equation}
which utilizes the information in the trash system. The diagonalization of trash system $B$ can provide us partial spectral information of $\rho_{AB}$, therefore we can choose appropriate $\{\epsilon_{j}\}$ to increase the mixedness of the output state $\widetilde \rho_{AB}$. From the information aspect, the noise-assisted quantum autoencoder extracts part of the information of $\widetilde \rho_{AB}$ through measurements, and then transfers the extracted classical information to quantum information via adding appropriate noise channels.

The fundamental reason that QAE can successfully compress a low-rank mixed state is that the encoder can decouple the correlation between subsystem $A$ and subsystem $B$. Our  method is a direct extension of this philosophy. Diagonalization of trash system $B$ sequentially in the computational basis can decouple the correlation between $A$, $B$ and the correlation within B, $i.e.$, only a minimum number of qubits in $B$ are still coupled with subsystem $A$ after encoding. Therefore, only a minimum number of qubits in $B$ will be operated in the noise channel. If the rank of input state $\rho_{AB}$, $k$, is no greater than $d_A$, the input state can be perfectly decoupled and compressed, each $\epsilon_{j}$ equals 0. If $k$ is greater than $d_A$, $\lceil \log_{2} \frac{k}{d_{A}} \rceil$ qubits in $B'$ will be added quantum noise, improving the rank of $\widetilde \rho_{AB}$ to $2^{\lceil \log_{2} \frac{k}{d_{A}} \rceil}d_{A}$. We note that the number of qubits are rounded up, therefore there exists information redundancy if $\log_{2} \frac{k}{d_{A}}$ is not an integer. This redundancy is inevitable since we only do local measurements and implement local noise channels on $B'$.

The fidelity of N-QAE is fully characterized by \begin{align}\label{eq:fid output QAE}
F(\rho^{out},\rho_A^{out}\otimes\rho_{B'}^{in}),
\end{align}
which is a mean-field-like approximation. Since the initial state is efficiently decoupled after encoding and the approximation is very precise, N-QAE can achieve high recovering fidelity. 

In experiments, we can prepare each $\sigma_{j}$ through amplitude-damping noise. Denote the relaxation time of the system as $T_1$, we initialize the system to $\proj 1$, after time $t_{j} = -T_{1}\ln (1-\epsilon_{j})$, we apply $X$ gate to the $j$th qubit, then we have the required $\rho_{j}$ since
\begin{align}
    \mathcal{N}_{j}(\proj 1) = \begin{pmatrix}
 1-e^{-t_{j}/T_{1}}& 0\\
 0& e^{-t_{j}/T_{1}} \end{pmatrix} 
 = X\sigma_{j}X.
\end{align}
Some other approaches, e.g. twirling operation, can also be used to prepare $\sigma_{j}$.

The workflow of N-QAE is as follows.
And we present the structure of 5-2-5 N-QAE as an example in Fig.~\ref{fig:NQAE}. 
	\begin{enumerate}
		\item Choose the ansatz of encoder unitary $U(\bm\theta)$ and initial parameters of $\bm\theta$.
		\item Apply the encoder $U(\bm\theta)$ to the initial state $\rho_{AB}$.
		
		\item Measure $\rho_{B}^{out}$ and update the cost function $L(\bm\theta)=  \tr[\rho_{B}^{out} H_{cost}]$.
		\item Perform optimization (e.g., gradient descent) of $L(\bm\theta)$ and obtain a new parameter $\bm\theta$.
		\item Repeat steps 2--4 until convergence with certain tolerance.
		\item Report the classical information $\bm\theta$, $f$,  and store the quantum state $\rho_A^{out}$.
		
		\item Implement the local amplitude-damping noise channel $\mathcal{N}$ to prepare $\rho_{B'}^{in}$.
		\item Apply $U(\bm\theta)^{\dagger}$ to $\rho_A^{out}\otimes\rho_{B'}^{in}$ and denote the output as $\widetilde\rho_{AB}$.
		\item Output the reconstructed state $\widetilde\rho_{AB}$.
	\end{enumerate}

\begin{figure}[H]
	\centering
	\includegraphics[width=8.5cm]{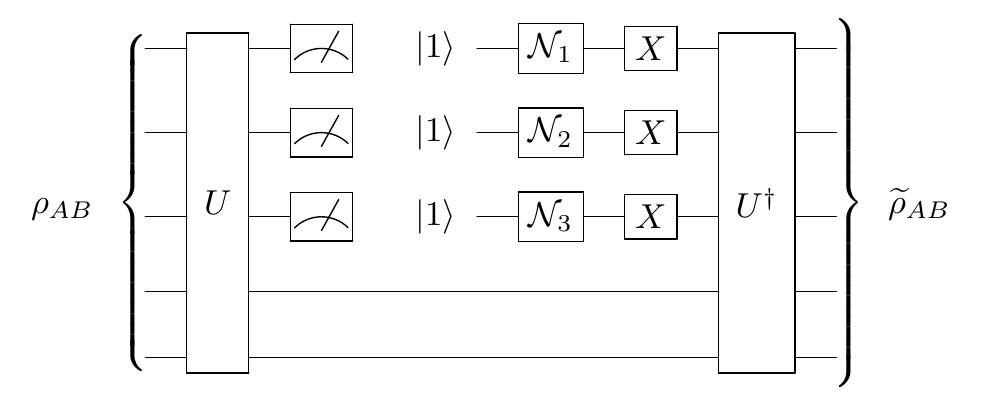}
	\caption{Structure of 5-2-5 N-QAE. $\rho_{B'}^{in}=\otimes_{j=1}^{n_{B}} \sigma_{j}$ is prepared by amplitude-damping noise channel $\mathcal{N} = \otimes_{j=1}^{n_{B}}\mathcal{N}_{j}$ and X gates.
	}
	\label{fig:NQAE} 
\end{figure}
Compared with the original QAE, N-QAE makes full use of the measured information of trash system $B$. Through diagonalization of $\rho_{B}^{out}$, N-QAE decouples the correlation between $A$, $B$ and the correlation within $B$. Replacing $\rho_{B}^{out}$ by a mixed product state $\rho_{B'}^{in}$, which is prepared by local noise channels, can efficiently keep the mixedness, therefore break through the rank limitation and improve the recovering fidelity.

\section{Adiabatic Quantum Autoencoder}\label{aqae}
Given Hamiltonian $H$, the thermal state of $H$ has the form
\begin{equation}
\rho_{\beta}=\frac{e^{-\beta H}}{\tr\left(e^{-\beta H}\right)},
\end{equation}
where $\beta = 1/k_{B}T$, $T$ is the temperature of the system.

Suppose we want to compress the thermal state $\rho_{\beta}$ of a given Hamiltonian $H$, QAE and NQAE can do this efficiently for small systems. However, with the increase of system size, the gradient vanishing problem, and local minimums in the optimization may lead to wrong results.  Here, we introduce an adiabatic model of quantum autoencoder (A-QAE) that can be implemented on quantum annealers to compress and recover $\rho_{\beta}$. Our model extends the applicable domain of the quantum adiabatic algorithm \cite{albash2018adiabatic} to mixed states.

We start the system in the thermal state $\rho_{\beta}$ and evolve according to the time-dependent Hamiltonian
\begin{equation}
H_{encoder}(t) = (1-\frac{t}{t_{a}}) H +\frac{t}{t_{a}}H_{l},
\end{equation}
$t_{a}$ is the annealing time, $H_l$ is the latent Hamiltonian. According to adiabatic theorem \cite{born1928beweis}, if the minimum gaps between the $j$th energy level and neighboring energy levels are strictly greater than 0 and the evolution is slow enough, the $j$th eigenstate of $H$, $\ket \psi _{j}$, will evolve to the $j$th eigenstate of $H_l$, $\ket \psi _{j}^l$ with high fidelity. 

In adiabatic quantum computing, we usually start from the ground state of the initial Hamiltonian, evolve the system Hamiltonian to the target Hamiltonian slowly, then we have the ground state of the target Hamiltonian for a gapped evolution path. In the A-QAE scheme, however, we start from a thermal state of the initial Hamiltonian, evolve the system to the latent Hamiltonian. If all neighboring gaps during the evolution are strictly greater than 0, the spectrum of the initial and the final state are the same, $i.e.$,

\begin{equation}\label{eq:aqae_spectrum}
    \sum_{j}p_j\proj{\psi_j}\overset{encoder}{\longrightarrow} \sum_{j}p_j\proj{\psi_j^l}
\end{equation}

In our scheme, we choose the latent Hamiltonian as
\begin{equation}
H_{l}=I_{A}\otimes(\mathbb{1}-\sum_{j=1}^{n_{B}} r_{j} Z_{j}),
\end{equation}
with $r_{j} = 1/2^{j-1}$. Each eigenstate of $H_l$ is a $d_{A}$-fold degenerate product state. The ideal adiabatic evolution in Eq.~\eqref{eq:aqae_spectrum} can fully decouple the correlation between subsystem $A$ and subsystem $B$.

The decoding evolution is the inverse of the encoding one. We set the initial state of the decoder to be $\rho_{A}^{out}\otimes \proj 0 _{B'}$, evolve the system according to the time-dependent Hamiltonian
\begin{equation}
H_{decoder}(t) = (1-\frac{t}{t_{a}}) H_{l} +\frac{t}{t_{a}}H.
\end{equation}
If the gap between the $d_{A}$th energy level and the $(d_{A}+1)$th energy level during the encoding and decoding evolutions is strictly greater than 0, A-QAE preserves the lowest $d_{A}$ eigenstates of $H$.

The same subroutine in noise-assisted quantum autoencoder (cf.~step 7) can be applied to improve the recovering fidelity. We extract the spectral information of $\rho_{B}^{out}$ via measurements on the trash $B$ system, then set the noise rates by
\begin{equation}
\epsilon_{j} = \frac{1}{2}-\frac{1}{2}\tr[\rho_{B}^{out} Z_{j}]
\end{equation}
and input
\begin{equation}
    \sigma_{j} = \begin{pmatrix}
 1-\epsilon_{j}& 0\\ 
 0& \epsilon_{j}
\end{pmatrix}
\end{equation}
to the $j$th qubit in system $B'$.
We call this modified model noise-assisted adiabatic quantum autoencoder. 

In summary, A-QAE and NA-QAE can decouple the correlation between $A$ and $B$ through an adiabatic evolution to a product latent Hamiltonian. Although unwanted level crossings may slightly lower the recovering fidelity, we do not need to train the encoder via measurements and optimization.

\section{Applications and experiments}\label{applications}
In this section, we numerically show some applications of our models. Each mixed state can be represented as the thermal state of a Hamiltonian, without loss of generality, we apply our noise-assisted algorithms to the thermal states of one-dimensional transverse-field Ising Model (TFIM) and Werner states.
We use parameterized quantum circuits to form the encoder and decoder. The structure is shown in Fig. \ref{fig:pqc}. Each block is repeated for $p$ times. We adaptively increase $p$ to find the appropriate depth. The first $R_Y$ and $R_Z$ rotations can prepare any product state;  CNOT entangling layers can introduce entanglement; $R_Z$-$R_Y$-$R_Z$ rotations between neighboring entangling layers can form any product unitary operator $\otimes_{j = 1}^{n_A + n_B}U_{j}$ up to a global phase. Our simulations and optimization loop are implemented via Paddle Quantum~\cite{paddlequantum} on the PaddlePaddle deep learning platform~\cite{Ma2019,paddlepaddle}
\begin{figure}[H]
	\centering
	\includegraphics[width=6cm]{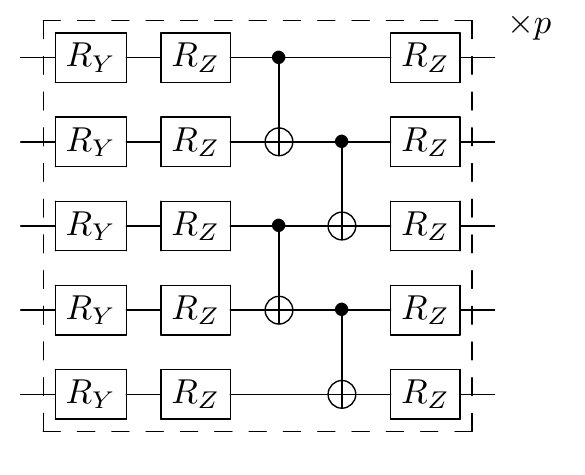}
	\caption{Example of the parameterized quantum circuit used in our experiments}
	\label{fig:pqc}
\end{figure}

\subsection{Transverse-field Ising model}
The Hamiltonian of one-dimensional transverse-field Ising model is
\begin{equation}
H^{tIsing} = -J(\sum_{j} Z_{j}Z_{j+1}  + g\sum_{j} X_{j}),
\end{equation}
in the following, we set $J = g = 1$, $n_{A} = 2$, $n_{B} = 3$. The system is in a gapless phase.

For $\rho_{\beta}$ with different $\beta$ values, we compress the thermal states 
\begin{equation}
    \rho_{\beta}=\frac{e^{-\beta H^{tIsing}}}{\tr\left(e^{-\beta H^{tIsing}}\right)} 
\end{equation}
with different models. In QAE and N-QAE, the optimizer is Adam, the learning rate is $0.05$. In A-QAE and NA-QAE, the latent Hamiltonian is $H_{l}=I_{A}\otimes(\mathbb{1}-\sum_{j=1}^{n_{B}} 1/2^{j-1} Z_{j})$.

The von Neumann entropy $\mathcal{S}(\rho)=-\tr(\rho\log\rho)$ is a measure of the mixedness of a given quantum state $\rho$. Here, the entropies of the input states and output states of different models are shown in Fig. \ref{fig:tfim4}. Noise can improve output mixedness to that of input. 

\begin{figure}[H]
	\centering
	\includegraphics[width=7.5cm]{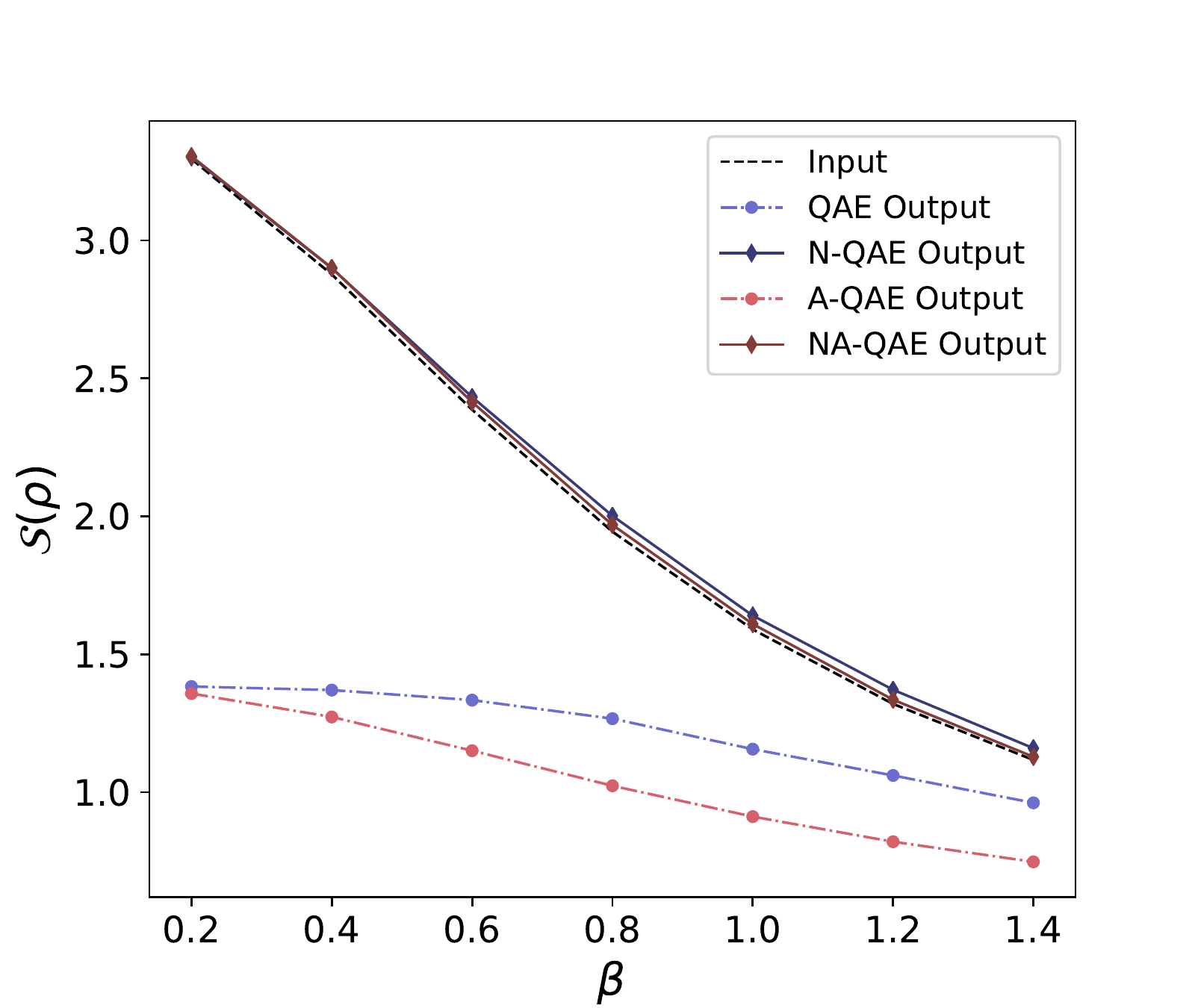}
	\caption{Input and output von Neumann entropy for compressing thermal states of TFIM with different models. The annealing time of A-QAE and NA-QAE is $t_{a} = 1 \times10^{3}$.}
	\label{fig:tfim4}
\end{figure}

The recovering fidelities are shown in Fig. \ref{fig:tfim}. For all $\beta$ values, noise-assisted quantum autoencoders can compress and recover the thermal state precisely. With the increase of $\beta$, the proportion of high-lying excited states in $\rho_{\beta}$ decreases exponentially, the required noise rates $\{\epsilon_{j}\}$ in our algorithms decreases accordingly. A-QAE performs worse than QAE; nevertheless, NA-QAE outperforms the gate-based N-QAE. 
In N-AQE, a deep parameterized encoder circuit will be required for large systems, and thus the optimization may be trapped to a local minimum. NA-QAE can efficiently diagonalize high-lying eigenstates. 

For the case when $\beta = 1$, the comparison between QAE and N-QAE for different iterations is shown in Fig. \ref{fig:tfim2}. N-QAE outperforms QAE dramatically after a critical point.

The comparison between A-QAE and NA-QAE for different annealing times is shown in Fig. \ref{fig:tfim3}, still, $\beta = 1$. We note that A-QAE recovering fidelity is far from the fidelity bound since our choice of $H_l$ and the evolution path is not optimal, some eigenstates may not be adiabatically connected to a product state. There are unwanted level crossings, Eq. (\ref{eq:aqae_spectrum}) is only approximately satisfied. For the A-QAE model, a better latent Hamiltonian is $H_{l} = I_{A}\otimes (-\sum_{j}Z_{j} -  \sum_{j}Z_{j}Z_{j+1})_{B}$, where the gap between the $d_A$th and the $(d_A + 1)$th energy level is much larger.

\begin{figure}[H]
\centering
\subfigure[]{
\label{fig:tfim}
\includegraphics[width=7.5cm]{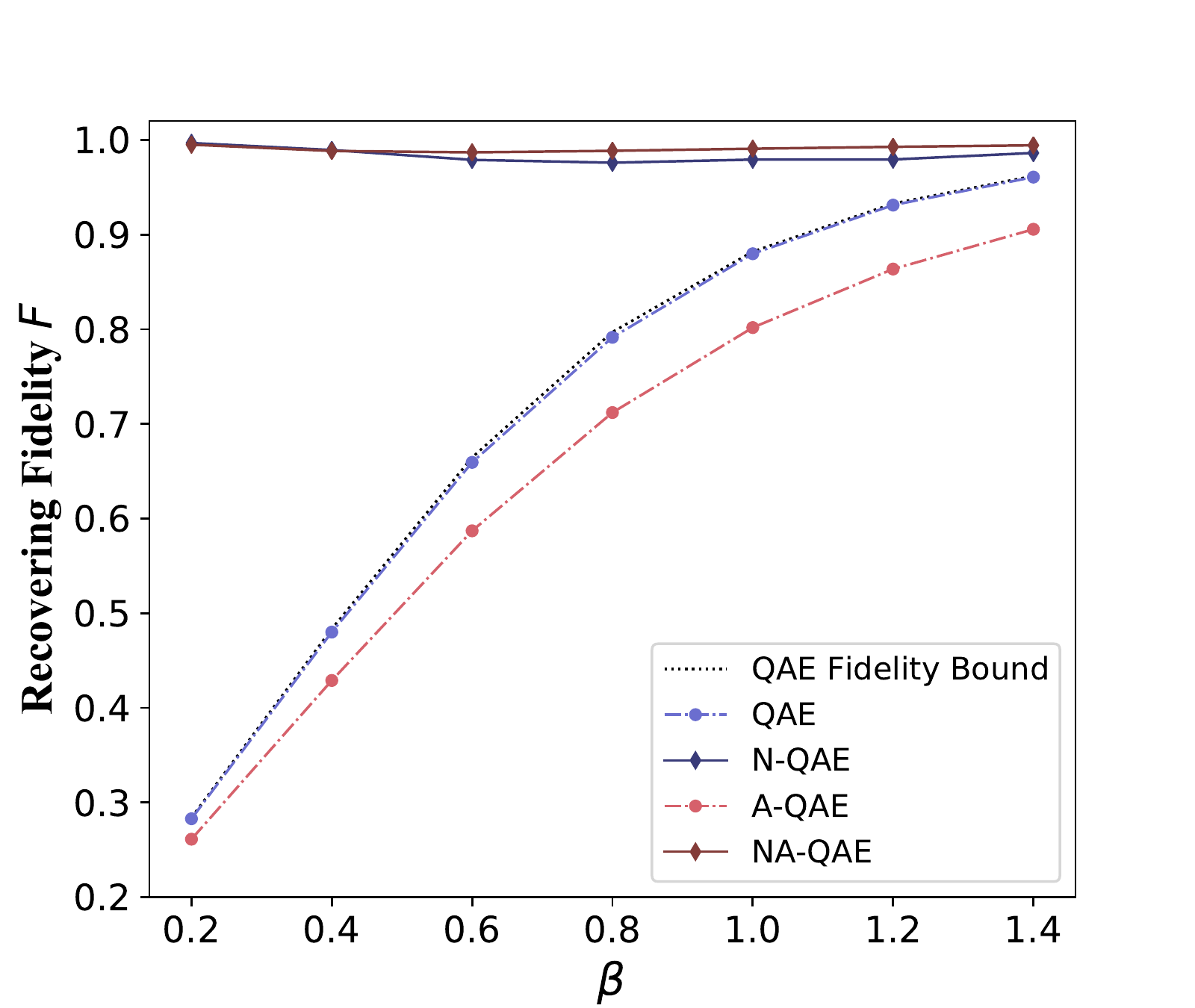}
}
\subfigure[]{
\label{fig:tfim2}
\includegraphics[width=7.5cm]{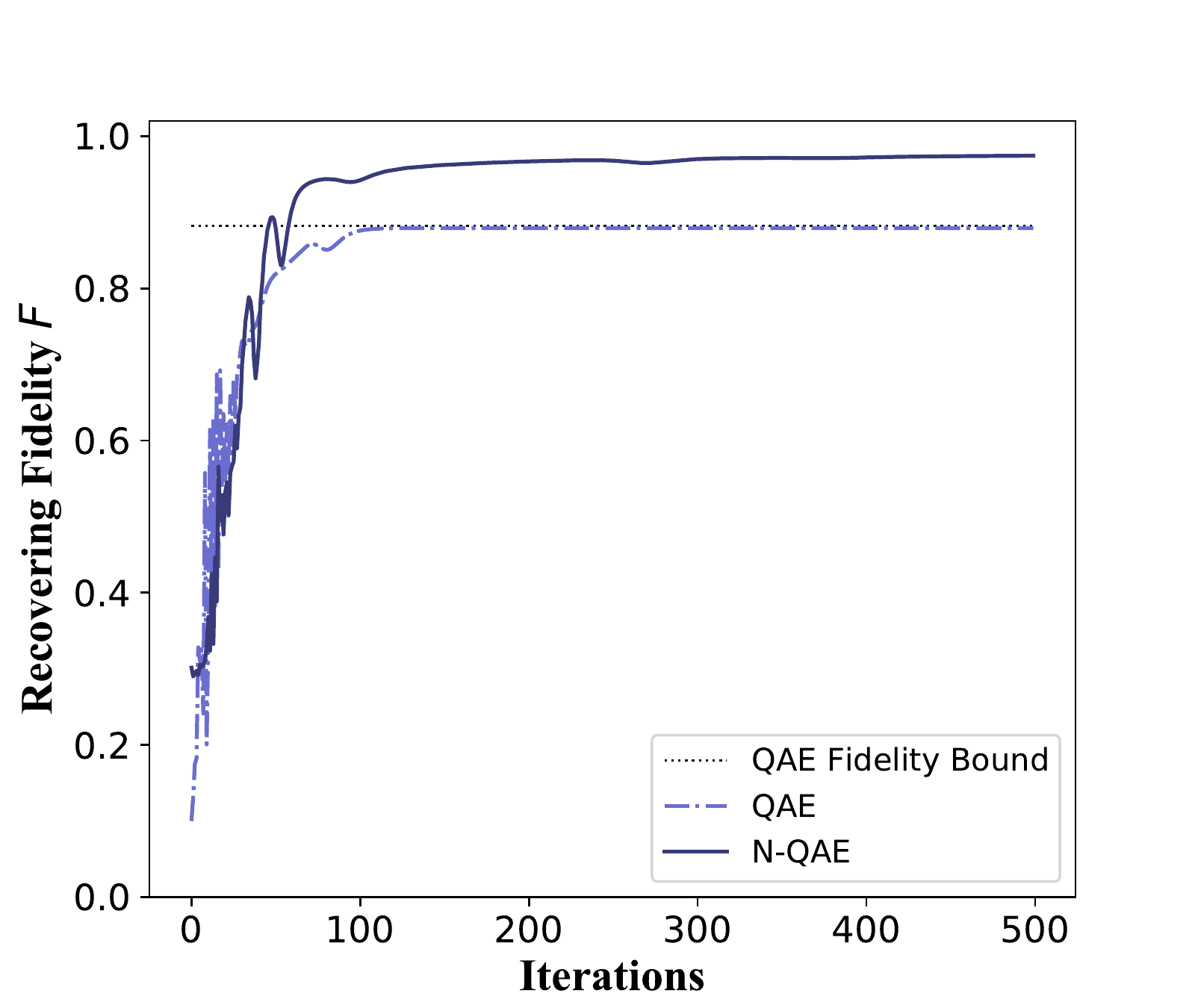}
}
\subfigure[]{
\label{fig:tfim3}
\includegraphics[width=7.5cm]{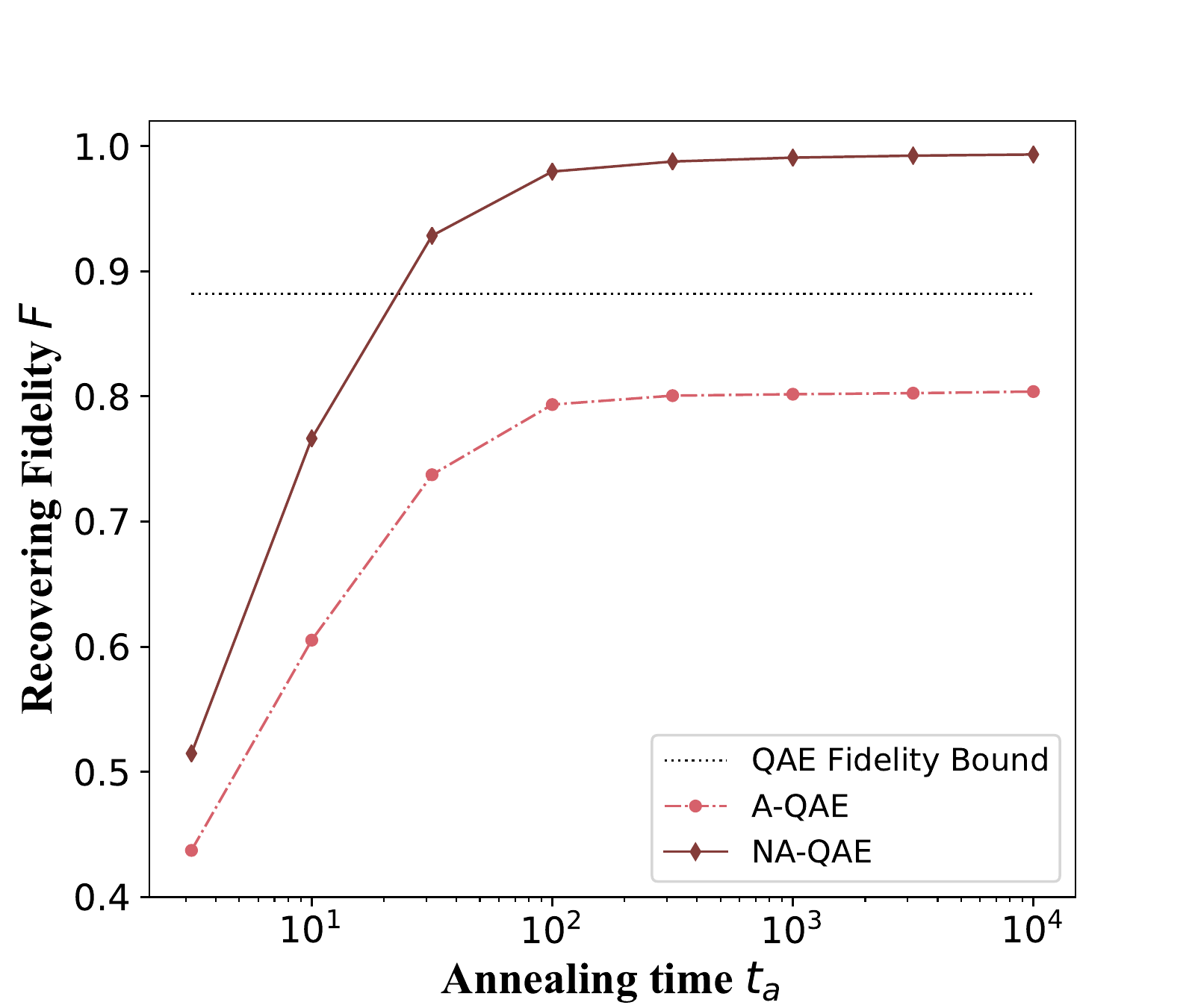}
}
\caption{Reconstruction fidelity for different models. Results for noiseless models are represented by dashed lines, results for noise-assisted models are represented by solid lines. (a) Recovering fidelity versus $\beta$. The annealing time of A-QAE and NA-QAE is $t_{a} = 1 \times10^{3}$. (b) Recovering fidelity versus the number of iterations for QAE and N-QAE. (c) Recovering fidelity versus annealing time $t_a$ for A-QAE and NA-QAE.  }
\end{figure}

\subsection{Werner states}
Werner state is a bipartite quantum state that is invariant under any unitary operator of the form $U \otimes U$, it can be parameterized by
\begin{equation}
    \rho_{W}(\alpha)=\frac{1}{d^{2}-d \alpha}\left(I_{d^{2}}-\alpha F\right),
\end{equation}
where $F=\sum_{kj}\ketbra{kj}{jk}$, $d$ is the dimension of each party of the state,  and $\alpha$ varies between $-1$ and $1$. In the following, we compress six-qubit Werner states with 6-5-6 QAE models, i.e. $d = 8$, $n_A = 5$, $n_B = 1$. Note that Werner state with six qubits is separable when $\alpha \leq 1/8$ and entangled when $\alpha > 1/8$.

In A-QAE and NA-QAE, we construct the Hamiltonian $H_{\alpha}$ such that $\rho_{W}(\alpha) = e^{-H_{\alpha}}$ for $-1 < \alpha < 1$, and set $H_l = H_{\alpha}$. The input and output von Neumann entropy of different models are shown in Fig. \ref{fig:werner_vn}, the recovering fidelities are shown in Fig. \ref{fig:werner_f}. Noise-assisted models can still enhance the mixedness and recover the initial state efficiently, nevertheless, they fail to increase the fidelity to 0.95 when $\alpha$ is close to -1. This indicates the limitation of N-QAE and NA-QAE for certain energy level-spacing distributions.

\begin{figure}[H]
	\centering
	\includegraphics[width=7.5cm]{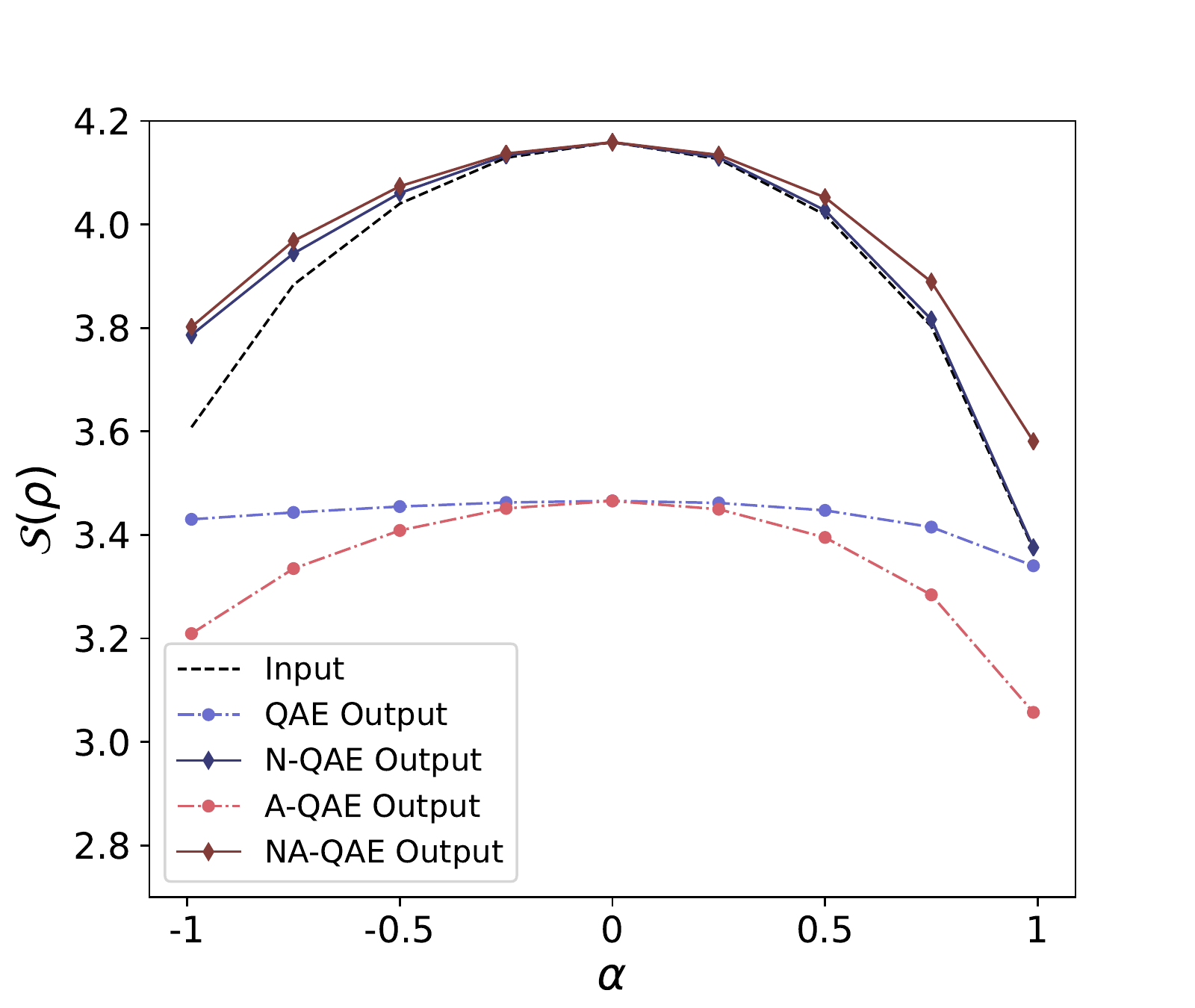}
	\caption{Input and output von Neumann entropy for compressing Werner states with different models. The annealing time of A-QAE and NA-QAE is $t_{a} = 1 \times10^{3}$.}
	\label{fig:werner_vn}
\end{figure}

Our models may also be applied to classical data compression. Take image compression as an example, given the matrix representation of an image, $G$. We can easily construct the corresponding Hamiltonian through
\begin{equation}
    H^{G} = (G + G^{T}) + i(G - G^{T}),
\end{equation}
prepare two thermal states $\rho^{G}_{\beta} = e^{-\beta H^{G}}/\tr(e^{-\beta H^{G}})$ and $\rho'^{G}_{\beta} = e^{-\beta (-H^{G})}/\tr(e^{-\beta (-H^{G})})$, which contain all the information of the image $G$. One could compress $\rho^{G}_{\beta}$ and $\rho'^{G}_{\beta}$ with quantum-autoencoder models, and then implement quantum tomography and extract the compressed classical information.
\begin{figure}[H]
	\centering
	\includegraphics[width=7.5cm]{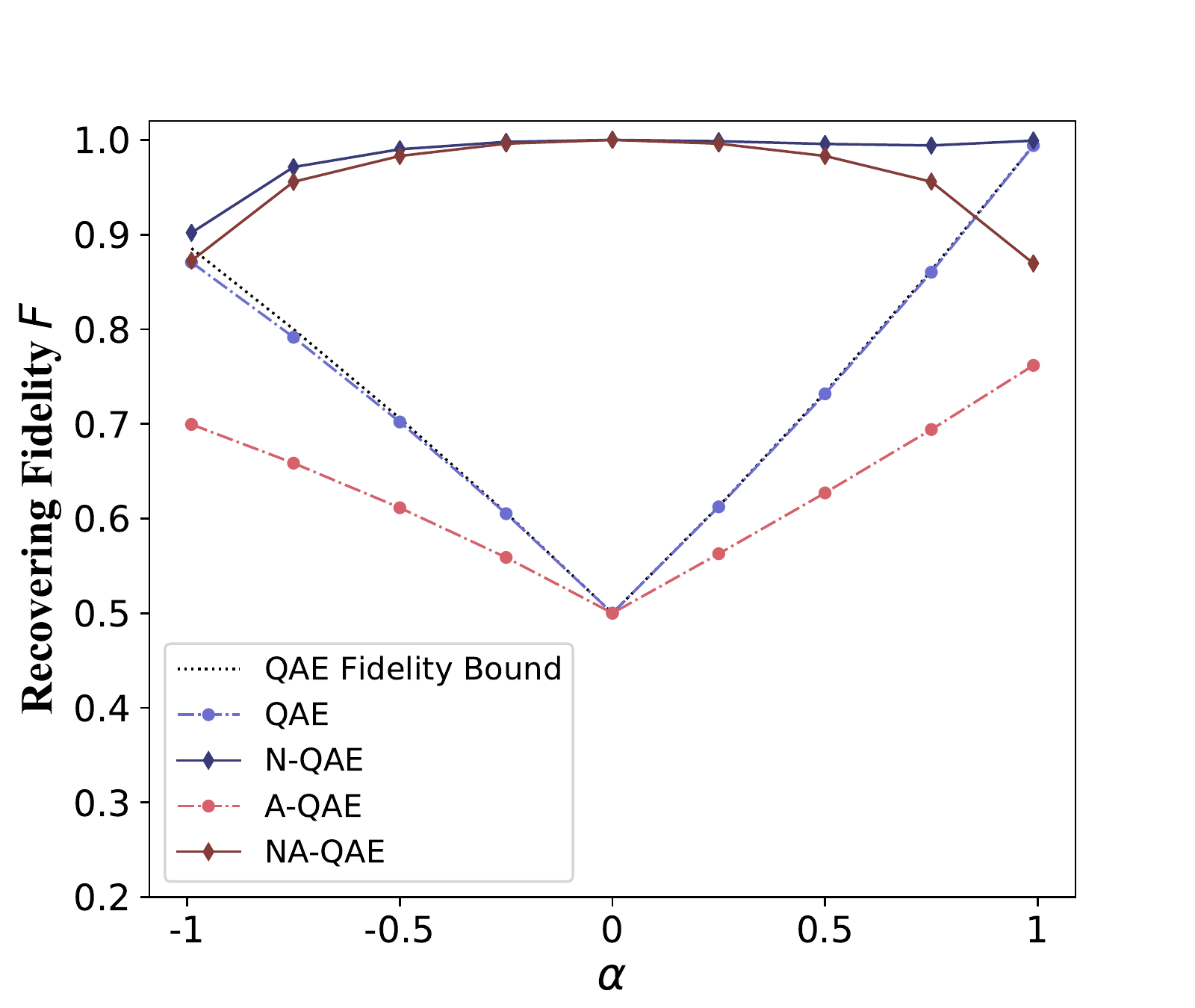}
	\caption{Recovering fidelity versus $\alpha$. The annealing time of A-QAE and NA-QAE is $t_{a} = 1 \times10^{3}$.}
	\label{fig:werner_f}
\end{figure}

\section{Conclusions and Outlook}\label{Conclusion}
We deliver an information-theoretic study on the limits of the QAE model and introduce noise-assisted QAE models that achieve a better reconstruction of quantum information. 
In standard QAE, an optimal encoder learns to decouple and compress eigenstates with the $d_A$ largest eigenvalues. We prove the fidelity bound of the standard QAE and its achievability. When the rank of input state $k$ is greater than $d_A$, additional eigenstate(s) are mapped to the space spanned by the aforementioned $d_A$ eigenstates. Therefore, one bottleneck of the standard QAE is the mixedness inconsistency between the input and the output, but our method could adjust it by adding suitable quantum operations, and we in particular utilize the information in the trash system. For mixed-state compression, we present a QAE variant assisted by noise, N-QAE, to achieve high-fidelity quantum data compression. For pure-state compression, we present P-QAE with a variational projection subroutine.

In addition to the optimization-based QAE, we introduce an adiabatic QAE model that can be run on quantum annealers. The noise-assisted version of A-QAE, NA-QAE, is also discussed. It will be interesting to study the optimal latent Hamiltonian for A-QAE and consider whether we can improve the recovering fidelity with a catalyst Hamiltonian~\cite{hormozi2017nonstoquastic, cao2021speedup}.

Numerical results of N-QAE and NA-QAE for compressing TFIM thermal states and Werner states demonstrate better performance. Our methods can be applied to general mixed-state compression, and the performance of N-QAE is no worse than the standard QAE in all cases.

Our models may also be applied to classical data compression. One possible way is to encode classical data to a specific Hamiltonian, prepare the Hamiltonian's thermal states and do compression, and then recover the initial information via quantum tomography. It is of interest to have further study on better schemes. Moreover, it will be interesting to further explore quantum information processing with the aid of machine learning, see Refs~\cite{Wallnofer2020,Zhao2021} for examples.

In addition to applying noise in a specific step, one may make use of gate noise in each operation as a resource directly. For example, a noisy quantum circuit without ancillary qubits may prepare the Gibbs state more efficiently than a noiseless quantum circuit. The positive aspects and applications of noise in the NISQ era are worth studying further. 

\acknowledgments
We thank Jiaqing Jiang, Runyao Duan, and Yinan Li for helpful discussions. This work was done when C. C. was a visiting student at Baidu Research.

\bibliography{references}

\appendix 
\section{Projected Quantum Autoencoder for An Ensemble of Pure States}\label{pure}
Sometimes, we may apply QAE to compress an ensemble of pure states instead of a single mixed state. When the dimension of the support of the ensemble is no greater than $d_A$, all states in the ensemble can be perfectly compressed and recovered since all states in that support Hilbert space can be perfectly decoupled. 

In this section, we modify the QAE model to improve its capability for quantum ensemble with high-dimensional support. The input state of the system $B'$ is determined by the measurement results of the "trash" system $B$ instead of $|0...0\rangle$. More specifically, we do projective measurements in the computational basis on system $B$, suppose $|m_{1}...m_{n_{B}} \rangle$ is the most probable string for input state $\ket\psi$, we input $|m_{1}...m_{n_{B}} \rangle$ as the initial state of $B'$.

Classical autoencoders can learn a representation of the set of data efficiently since they can implement nonlinear transformations and therefore extract complex features. Our P-QAE model is more capable than the model in Ref. \cite{Romero2017} due to the nonlinearity from measurements. 

The original QAE model can exactly recover states in one subspace $\mathcal{H}_{1}$ with $\operatorname{dim} \mathcal{H}_{1} = d_{A}$, $i.e.$, if each state in $\{\ket \psi\}$ satisfies $\ket \psi \in \mathcal{H}_{1}$, the ensemble can be exactly recovered by the standard QAE. Our modified QAE model can exactly recover states in $d_B$ such subspaces. Divide the input $d_{A}d_{B}-$dimensional Hilbert space $\mathcal{H}$ to $d_{B}$ orthogonal subspaces $\{\mathcal{H}_{j}\}_{j=1}^{d_B}$, with all $\operatorname{dim} \mathcal{H}_{j} = d_{A}$,
\begin{equation}
    \mathcal{H} = \bigoplus_{j=1}^{d_{B}}\mathcal{H}_{j}
\end{equation}
if each state $\ket \psi$ is in one of the subspaces $\mathcal{H}_j$, $\ket \psi \in \mathcal{H}_{j}$, the ensemble $\{\ket \psi\}$ can be exactly recovered by P-QAE.

We train the encoder by projective measurements on the system $B$. The loss function is defined by
\begin{equation}
L_{p} = 1 - \sum_{j = 1}^{n_{B}}(\tr(\rho_{B}^{out} Z_{j}))^{2},
\end{equation}
where $\rho_{B}$ is the state of system $B$. If we can optimize the loss function to 0, all input states can be perfectly recovered. 

After training, we variationally project $\rho_{A}^{out}$ to its eigenstate with the largest eigenvalue and input it to the decoder. The recovering fidelity can be further improved because the rank of the recovered state is consistent with the initial state, $i.e.$, $1$. 

For example, consider pure state $U\ket\psi = \sqrt{\frac{2}{3}}\ket{00} + \sqrt{\frac{1}{3}}\ket{11}$, then
\begin{align}
F(\widetilde\psi,\psi) & = \tr \proj\psi U^\dagger(\rho_{A}^{out}\otimes\proj0_B)U\\
& =  \tr U\proj\psi U^\dagger(\rho_{A}^{out}\otimes\proj0_B)\\
& = \frac{4}{9}. \label{eq:example original fidelity}
\end{align}
If we add a subroutine $\cP$ to project the compressed state $\psi_A^{out}$ to its eigenstate with largest eigenvalue, $i.e.$ $\proj 0_A$. The fidelity of reconstruction in this case is 
\begin{align}
F(\widetilde\psi,\psi) = \tr \proj\psi U^\dagger(\proj 0_A\otimes\proj0_B) U = \frac23,
\end{align}
which is strictly better than the fidelity in Eq.~\eqref{eq:example original fidelity} obtained via the previous approach. 

\begin{figure}[H]
	\centering
	\includegraphics[width=8.5cm]{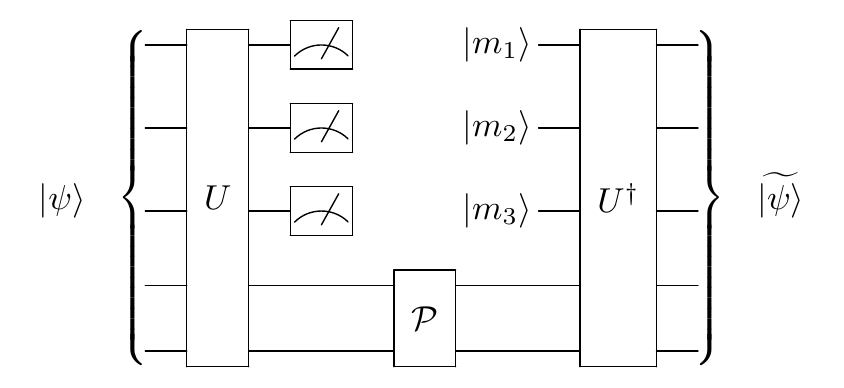}
	\caption{Structure of 5-2-5 P-QAE}
	\label{fig:pqae1}
\end{figure}

The hybrid algorithm of the projected quantum autoencoder is as follows (see Fig.~\ref{fig:pqae1} for the general scheme).
\begin{enumerate}
\item Training the encoder for the ensemble
\begin{enumerate}
	\item Choose the ansatz of encoder unitary $U(\bm\theta)$ and initial parameters of $\bm\theta$;
	\item Apply the encoder $U(\bm\theta)$ to the initial state $\psi_{AB}$;
	
	\item Measure the trash system $B$ and update cost function $L_{p} = 1 - \sum_{j = 1}^{n_{B}}(\tr(\rho_{B}^{out} Z_{j}))^{2}$;
	
	\item Perform optimization of $L_p(\bm\theta)$ and update parameter $\bm\theta$;
 
	\item Repeat steps b--d until convergence with certain tolerance;
	
	\item Report the classical information $\bm\theta$ and the most probable string of $\rho_{B}^{out}$, $\ket m$.
\end{enumerate}	
\item Training the decoder for recovering the state $\ket \psi$ in the ensemble
    \begin{enumerate}
	\item Apply $U(\bm\theta)$ to $\ket \psi$ and keep the compressed state $\rho_A^{out} = \tr_B U(\bm\theta)\proj \psi U(\bm\theta)^\dagger$;
   \item Apply the projection subroutine $\cP$ to $\rho_A^{out}$ and obtain a pure state $\ket v$;
	
	\item Apply $U(\bm\theta)^{\dagger}$ to $\ket v_A\otimes \ket{m}_{B}$ and obtain $\ket{\widetilde{\psi}} = U(\bm\theta)^{\dagger} \ket v_A \otimes \ket m_B$;
	
	\item Output the reconstructed state $\ket{\widetilde{\psi}}_{AB}$.
    \end{enumerate}
\end{enumerate}

The projection subroutine $\cP$ can be implemented as follows:
\begin{enumerate}
 \item Choose the ansatz of the unitary $W(\bm\beta)$ with initial $\bm\beta$;
	\item Apply the unitary $W(\bm\beta)$ to the compressed state $\psi_A^{out}$;	
	\item Measure the overlap $f_A=\bra 0_AW(\bm\beta)\rho_A^{out}W(\bm\beta)^\dagger \ket 0_A$  and update $L_A(\bm\beta)= 1- f_A$;
	
	\item Perform optimization of $L_A(\bm\beta)$ and update parameter $\bm\beta$;
	
	\item Repeat steps 2--4 until convergence with certain tolerance;
	\item Output the state $W(\bm\beta)^\dagger \ket 0_A$
\end{enumerate}

The fidelity of reconstruction of each state via QAE is fully characterized by 
\begin{align}\label{eq:fid output QAE}
F(\psi^{out},\rho_A^{out}\otimes\proj0),
\end{align}
since
$F(\widetilde\psi,\psi)  = \tr \proj\psi U^\dagger(\rho_{A}^{out}\otimes\proj0_B)U
 = \tr \proj{\psi^{out}}(\rho_{A}^{out}\otimes\proj0_B)$.

Similarly, the fidelity of reconstruction of each state via PQAE is given by 
\begin{align}\label{eq:fid output PQAE}
F(\psi^{out},\cP(\rho_A^{out})\otimes\proj m)
\end{align}
with $\cP(\rho_A^{out})$ gives the eigenstate of $\rho_A^{out}$ with maximum eigenvalue, $\ket m$ being the most probable string of $\rho_{B}^{out}$ in the computational basis.

For QAE and PQAE with $n_{A} = 2$, $n_{B}=3$, we randomly sample an ensemble of $N$ pure quantum states from the uniform Haar measure, $\{\ket \psi_{j}\}_{j=1,...,N}$, then train QAE and PQAE to compress and recover the quantum ensemble, the recovering fidelities are shown in Fig. \ref{fig:pqae2}.

\begin{figure}[H]
	\centering
	\includegraphics[width=7.5cm]{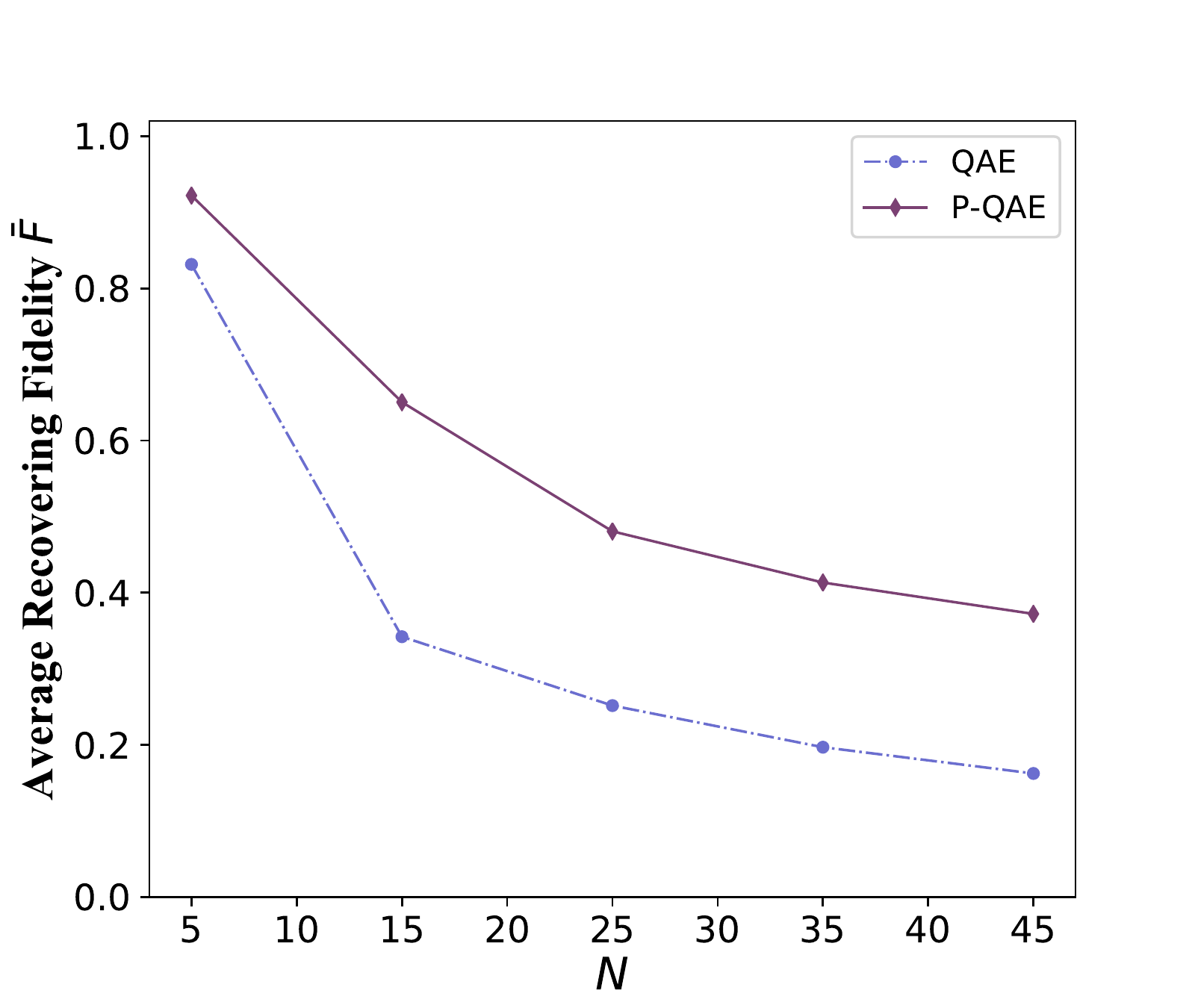}
	\caption{Recovering fidelity versus the number of random pure states in the ensemble. P-QAE can recover the ensemble with higher average fidelity than original QAE due to more recovering subspaces and rank-consistence. 
	}
	\label{fig:pqae2}
\end{figure}

Similar to N-QAE and NA-QAE, the target of P-QAE is to make input and output mixedness consistent, but the method is opposite since all inputs are pure states. On the one hand, we apply a measurement projection on the trash system $B$, compress each qubit information to a classical bit information, on the other hand, we use variational projection to extract the eigenstate of $\rho_{A}^{out}$ with the largest eigenvalue. The output of P-QAE is always a pure state due to the projections, the recovering fidelity can be improved for general pure inputs.

\end{document}